\documentclass[12pt]{article}
\usepackage{amsthm,amsfonts,amsmath,amssymb,amscd, accents, mathrsfs}
\usepackage{authblk}
\usepackage{hyperref}

                                \usepackage{verbatim}
\swapnumbers
                              
\theoremstyle{plain}

\usepackage{enumitem}
\numberwithin{equation}{section}
\newtheorem{theorem}{Theorem}[section]

\theoremstyle{definition}
\newtheorem{definition}[theorem]{Definition}
\newtheorem{example}[theorem]{Example}

\theoremstyle{remark}

\numberwithin{equation}{section}

\newcommand{\bR}{{\mathbb R}}

\newcommand{\cB}{{\mathcal B}}

\newcommand{\cH}{{\mathcal H}}

\newcommand{\cN}{{\mathcal N}}

\newcommand{\cS}{{\mathcal S}}

\newcommand{\ket}[1]{\left\vert #1\right\rangle}
\newcommand{\bra}[1]{\left\langle #1\right\vert}

\renewcommand{\Re}{\,\mathrm{Re}\,}   
\renewcommand{\Im}{\,\mathrm{Im}\,}   
\def\idty{{\mathchoice {\mathrm{1\mskip-4mu l}} {\mathrm{1\mskip-4mu l}} %
{\mathrm{1\mskip-4.5mu l}} {\mathrm{1\mskip-5mu l}}}}

\newcommand{\Tr}{\mathrm{Tr}}

\newcommand{\be}{\begin{equation}}
\newcommand{\ee}{\end{equation}}
\newcommand{\bea}{\begin{eqnarray}}
\newcommand{\eea}{\end{eqnarray}}
\newcommand{\beann}{\begin{eqnarray*}}
\newcommand{\eeann}{\end{eqnarray*}}

\usepackage{color}

\begin{document}

\title{Closest separable state when measured by a quasi-relative entropy}
\author{Anna Vershynina}
\affil{\small{Department of Mathematics, Philip Guthrie Hoffman Hall, University of Houston, 
3551 Cullen Blvd., Houston, TX 77204-3008, USA}}
\renewcommand\Authands{ and }
\renewcommand\Affilfont{\itshape\small}

\date{\today}

\maketitle

\begin{abstract} It is well known that for pure states the relative entropy of entanglement is equal to the reduced entropy, and the closest separable state is explicitly known as well. The same holds for Renyi relative entropy per recent results. We ask the same question for a quasi-relative entropy of entanglement, which is an entanglement measure defined as the minimum distance to the set of separable state, when the distance is measured by the quasi-relative entropy. First, we consider a maximally entangled state, and show that the closest separable state is the same for any quasi-relative entropy as for the relative entropy of entanglement. Then, we show that this also holds for a certain class of functions and any pure state. And at last, we consider any pure state on two qubit systems and a large class of operator convex function. For these, we find the closest separable state, which may not be the same one as for the relative entropy of entanglement. 

\end{abstract}

\section{Introduction}

Entanglement describes strong quantum ``connections" (stronger than any classical ones) that parts of a system can posses even when separated by long distance and having no connecting matter in between. It found applications in quantum optics \cite{BKAB00, LZYZ03}, nuclear magnetic resonance \cite{NKL98, PZFFLG03}, condensed mater physics \cite{FKR04, LSADSS07}, quantum cryptography \cite{DCGZ03, LWGZ08, SWGZ09}, and quantum algorithms and protocols \cite{BB84, BBCJPW93}. Understanding the behavior of entanglement in a system, its creation and loss (decoherence), is one of the fundamental goals of quantum science. For an overview on entanglement see \cite{4H09, PV14}.

The amount of entanglement in a quantum state can be captured using entanglement measures. Any such measure, $E(\rho)$ of a quantum state $\rho$, should fulfill at least the following criteria \cite{VP98}:
\begin{enumerate}
\item[(E1)] $E(\rho)\geq 0$, and $E(\rho)=0$ if and only if $\rho$ is separable;
\item[(E2)] Invariance under local unitary operations: $E(\rho)=E(U_A\otimes U_B\rho U_A^*\otimes U_B^*)$.
\item[(E3)] $E$ cannot increase under local operations and classical communication (LOCC), i.e. $E(\Lambda(\rho))\leq E(\rho)$ for any LOCC map $\Lambda$.
\end{enumerate}
It's interesting to note that these conditions are not set in stone for all entanglement measures. They can be relaxed, strengthen or new ones added depending on specific circumstances one would like to consider. For example, even the simples condition (E1) is not satisfied for a distillable entanglement, as it is zero on the bound entangled states \cite{3H98}, which are entangled. However, the measure is nice enough for us to overlook that lapse, explaining that there are various kinds of entanglement, all of which cannot be captured by one single measure, yet. New conditions on entanglement measures sometimes take the form of normalization on the maximally entangled states, convexity, additivity, continuity, and reduction to the marginal state on pure states \cite{3H00, VPRK97}.

A large class of entanglement measures is defined using a 'distance' between two quantum states \cite{VP98, VPRK97}. For a distance measure $D(\rho\|\sigma)$ between two quantum states $\rho$ and $\sigma$, which is not required to be a metric, entanglement measure is defined as
$$E(\rho):=\min_{\sigma\in\cS}D(\rho\|\sigma)\ , $$
where $\cS$ is the set of separable states. 

This construction gives entanglement measures such as: relative entropy of entanglement \cite{VP98}, Bures measure of entanglement \cite{VP98}, several R\'enyi relative entropies of entanglement \cite{TUU18, V19, ZHC17}, Tsallis relative entropy of entanglement \cite{V19}, Groverian measure of entanglement \cite{BNO02, SSB06, SKB10}, and possibly others.

In \cite{VP98, VPRK97} a set of conditions on the distance was given in order to induce an entanglement measure that satisfies conditions (E1)-(E3). These conditions are:
\begin{enumerate}
\item[(D1)]  $D(\rho\|\sigma)=0$ if and only if $\rho=\sigma$.
\item[(D2)] Invariance under unitary operations: $D(\rho\|\sigma)=D(U\rho U^*\|U\sigma U^*)$.
\item[(D3)] Data processing inequality (or monotonicity under completely positive trace-preserving (CPTP) maps): $D(\cN(\rho)\|\cN(\sigma))\leq D(\rho\|\sigma)$, for any CPTP map $\cN$.
\end{enumerate}

In \cite{VPRK97} it was shown that a relative entropy was a good candidate for a distance, since it induces an entanglement measure that satisfies (E1)-(E3). The relative entropy is defined as
$$S(\rho\|\sigma)=\Tr(\rho\log\rho)-\Tr(\rho\log\sigma) \ .$$
It was also shown in \cite{VPRK97} that a relative entropy of entanglement reduces to the marginal entropy on a maximally entangled state, and later in \cite{VP98} it was generalized to all pure states: for a pure state $\ket{\Psi}=\sum_j\sqrt{p_j}\ket{jj}$, the closest separable state when measured by the relative entropy is $\sigma^*=\sum_jp_j\ket{jj}\bra{jj}$ and the entropy of entanglement is $E_r(\ket{\Psi}\bra{\Psi})=-\sum_jp_j\log p_j$.
Sometimes this property of the relative entropy is elevated to a condition on an entanglement measure:
\begin{itemize}
\item[(E4)] Entanglement of a pure state $\rho=\ket{\Psi}\bra{\Psi}$ is the marginal entropy, i.e. $E(\ket{\Psi})=S(\rho_A)$, where $\rho_A=\Tr_B\ket{\Psi}\bra{\Psi}$, and $S(\rho)=-\Tr(\rho\log\rho)$ is the quantum entropy.
\end{itemize}

Similar results were found for the Renyi entropy of entanglement, defined by the minimal distance to the space of separable states when the distance is taken to be the Renyi relative entropy
$$S^R_\alpha(\rho\|\sigma)= \frac{1}{\alpha-1}\log\Tr(\rho^\alpha\sigma^{1-\alpha})\ .$$
An explicit expression was found for the  Renyi entropy of entanglement in \cite{CVGG20, ZHC17}. In \cite{CVGG20}, it was shown that for a pure state $\ket{\Psi}=\sum_j\sqrt{p_j}\ket{jj}$ and $\alpha<2$, the closest separable state is 
$$\sigma=\frac{1}{P_\alpha}\sum_j p_j^{1/\alpha} \ket{jj}\bra{jj}\ ,$$
where $P_\alpha=\sum_j p_j^{1/\alpha}.$ The Renyi entropy of entanglement is then given by
$$E_\alpha^R(\Psi)=\frac{\alpha}{\alpha-1}\log\sum_j p_j^{1/\alpha}\ . $$ 

We consider a large class of distances, called quasi-relative entropy, defined as
$$S_f(\rho|| \sigma)=\Tr(f(\Delta_{\sigma,\rho}){\rho})\ ,$$
where the relative modular operator, $\Delta_{A,B}(X)=L_AR_B^{-1}(X)=AXB^{-1}$. This class includes the regular relative entropy, the exponential Renyi relative entropy, and scaled Tsallis relative entropy.

Since the quasi-relative entropy satisfies conditions (D1)-(D3), the entanglement measure it induces, satisfy (E1)-(E3). The question now is what about (E4)? Or, what is the closest separable state to a pure entangled state when measured by the quasi-relative entropy?

This optimization question was asked in \cite{GGF14}, where the authors developed a method of solving this problem for any convex function on the set of quantum states. While authors considered the case of quasi-relative entropy as an example, they focused more on developing a general method rather than giving an explicit solution, as we aim in this manuscript. The method we are using is similar, if not the same, to the method in \cite{GGF14}, but it could be simplified and traced back to the original relative entropy of entanglement paper \cite{VP98}.

In Theorem \ref{thm:max} we consider a maximally entangled (Bell) state $\ket{\Psi^+}=\sum\frac{1}{\sqrt{d}}\sum_j\ket{jj}$, and show that the closest state is $\sigma^*=\sum\frac{1}{d}\sum_j\ket{jj}\bra{jj}$, which means that the quasi-relative entropy of entanglement is equal to $$E_f(\ket{\Psi^+})=f(1/d)\ .$$

In Theorem \ref{Entropy} we consider any pure state, and a class of functions $f$ that define quasi-relative entropy. In these conditions, we show that the closest separable state to a pure state $\ket{\Psi}=\sum_j\sqrt{p_j}\ket{jj}$ is the state $\sigma^*=\sum_jp_j\ket{jj}\bra{jj}$. The quasi-relative entropy of entanglement is then $$E_f(\ket{\Psi})=\sum_jp_jf(p_j)\ .$$

In Theorem \ref{thm:qubit} we take any qubit state $\ket{\Psi}=\sqrt{p}\ket{00}+\sqrt{1-p}\ket{11}$ and a large class of functions. We show that the closest separable state  is  in a form $\sigma^*=q\ket{00}\bra{00}+(1-q)\ket{11}\bra{11}$. We explicitly define these $q_j$'s in the proof of the theorem. In this case, the quasi-relative entropy becomes
$$E_f(\ket{\Psi}\bra{\Psi})=\sum_jp_jf(q_j)\ . $$

\section{Preliminaries (for more details see \cite{B97, CV-18, V19-2, V19-3})}
\subsection{Operator Monotone Functions}
Denote the set of operator monotone decreasing functions $f$ as $\mathcal{Q}_{(0,\infty)}$.

\begin{example}
From \cite[Exercise V.4.8]{B97} The following functions belong to $\mathcal{Q}_{(0,\infty)}$: 
\begin{itemize}
\item $f(x)=-\log x$, 
\item $f(x)=-x^p$ for $p\in[0,1]$, 
\item $f(x)=x^p$ for $p\in[-1,0]$.
\end{itemize}
\end{example}

According to \cite[Chapter II, Theorem I]{Dono} every function  $f\in\mathcal{Q}_{(0,\infty)}$, has a canonical integral representation 
\begin{equation}\label{low}
f(x) = -a_f x- b_f +\int_{0}^\infty \left(  \frac{1}{t +x }-\frac{t}{t^2+1}   \right){\rm d}\mu_f(t)\ ,
\end{equation}
where  $a_f:=-\lim_{y\uparrow\infty}\frac{f(iy)}{iy}\geq 0$, $b_f:=-\Re f(i)\in\bR$ and $\mu_f$ is a positive measure on $(0,\infty)$ such that 
${\displaystyle \int_{0}^\infty  \frac{1}{t^2+1}{\rm d}\mu_f(t)<\infty}$, and
 \begin{equation}\label{muform}
 \mu_f(x_1) - \mu_f(x_0) = -\lim_{y\downarrow 0} \frac{1}{\pi} \int_{x_0}^{x_1} \Im f(-x+iy){\rm d} x\ .
 \end{equation}
Conversely, every such function belongs in $\mathcal{Q}_{(0,\infty)}$.

We consider functions $f\in\mathcal{Q}_{(0,\infty)}$ such that $f(1)=0$. The last condition is equivalent to
$$0=f(1) = -a_f - b_f +\int_{0}^\infty \left(  \frac{1}{t +1 }-\frac{t}{t^2+1}   \right){\rm d}\mu_f(t)\ ,
 $$
 in other words,
 \begin{equation}\label{eq:f10}
 a_f+b_f=\int_{0}^\infty \left(  \frac{1}{t +1 }-\frac{t}{t^2+1}   \right){\rm d}\mu_f(t)\ .
 \end{equation}
 Therefore, the operator monotone decreasing function $f$ such that $f(1)=0$ has the following integral representation
 \begin{equation}\label{low2}
f(x) = a_f(1-x)+ \int_{0}^\infty \left(  \frac{1}{t +x }- \frac{1}{t +1 }\right){\rm d}\mu_f(t)\ .
\end{equation}

\begin{example}\label{ex-power}
Consider the power function $f(x)=-x^p$ for $p\in(0,1)$. It is operator monotone decreasing. Then 
$$a_f = -\lim_{y\uparrow\infty}f(iy)/(iy) = 0\ , \ \ \text{and } \ b_f = {\cos}(p \pi/2)\ .$$
For $x>0$, $\lim_{y\downarrow 0} \Im f(-x + iy) = -x^p\sin(p \pi)$ so that
$${\rm d}\mu_f(x) = \pi^{-1}\sin(p \pi) x^p{\rm d}x\ .$$ This yields the representation
\begin{equation}\label{eq:powex}
-x^p =  -{\cos}(p \pi/2)  + \frac{\sin(p \pi)}{\pi} \int_{0}^\infty t^p\left(  
\frac{1}{t +x } -\frac{t}{t^2+1}  \right){\rm d}t \ .
\end{equation}

\begin{example}\label{ex-log} Let $f(x)=-\log(x)$. It is operator monotone decreasing.
Then 
$$b_f=\Re(\log(i)) = 0\ ,$$ and 
$$a_f=\lim_{y\uparrow\infty}\log(iy)/(iy) = \lim_{y\uparrow\infty}(\log y + i\pi/2)/(iy) =0\ . $$ It is clear from \eqref{muform} that
$${\rm d}\mu_f(x) = \frac{1}{\pi}\lim_{y\downarrow 0}\Im\log(-x + iy){\rm d}x = {\rm d}x\ .$$
Then the  integral representation (\ref{low}) gives the following formula for the logarithmic function
\begin{equation}\label{eq:logex}
-\log x=\int_{0}^\infty \left( \frac{1}{t +x } - \frac{t}{t^2+1}  \right){\rm d}t \ ,
\end{equation}
which is also obvious from the direct computation of the integral.
\end{example}
\end{example}

\subsection{Quasi-relative entropy}
Quantum quasi-relative entropy was introduced by Petz \cite{P85, P86} as a quantum generalization of a classical Csisz\'ar's $f$-divergence \cite{C67-2}. It is defined in the context of von Neumann algebras, but we  consider only the Hilbert space setup. Let $\cH$ be a  finite-dimensional  Hilbert space, and $\rho$ and $\sigma$ be two states (given by density operators).

\begin{definition}\label{def:qre}
For an operator monotone decreasing function $f$, such that $f(1)=0$, and strictly positive states $\rho$ and $\sigma$ acting on a finite-dimensional Hilbert space $\cH$, {\it the quasi-relative entropy} (or sometimes referred to as {\it the $f$-divergence}) is defined as 
$$S_f(\rho|| \sigma)=\Tr(f(\Delta_{\sigma,\rho}){\rho})\ ,$$
where the relative modular operator, introduced by Araki \cite{A76}, $$\Delta_{A,B}(X)=L_AR_B^{-1}(X)=AXB^{-1}$$ is a product of left and right multiplication operators, $L_A(X)=AX$ and $R_B(X)=XB$. Throughout this paper we consider finite-dimensional setup, so the operators are invertible. (In general, $A^{-1}$ is stands for the generalized inverse of $A$.)
\end{definition}

There is a straightforward way to calculate the quasi-relative entropy from the spectral decomposition of states. Let $\rho$ and $\sigma$ have the following spectral decomposition
\begin{equation}\label{eq:spectral}
\rho=\sum_j\lambda_j\ket{\phi_j}\bra{\phi_j}, \ \ \sigma=\sum_k\mu_k\ket{\psi_k}\bra{\psi_k}\ ,
\end{equation}
where the eigenvalues are ordered: 
$$\lambda_n\leq \dots\leq \lambda_1, \ \ \mu_n\leq \dots\leq \mu_1\ . $$
 the set $\{\ket{\phi_k}\bra{\psi_j}\}_{j,k}$ forms an orthonormal basis of $\cB(\cH)$, the space of bounded linear operators, with respect to the Hilbert-Schmidt inner product defined as $\langle A, B \rangle=\Tr(A^*B)$. By \cite{V16}, the modular operator can be written as
\begin{equation}\label{eq:modular}
\Delta_{\sigma, \rho}=\sum_{j,k} \frac{\mu_k}{\lambda_j}P_{j,k}\ ,
\end{equation}
where $P_{j,k}:\cB(\cH)\rightarrow\cB(\cH)$ is defined by
$$P_{j,k}(X)=\ket{\psi_k}\bra{\phi_j}\bra{\psi_k}X\ket{\phi_j}\ . $$
The quasi-relative entropy is calculated as follows
\begin{equation}\label{eq:formula}
S_f(\rho||\sigma)=\sum_{j,k}\lambda_j f\left(\frac{\mu_k}{\lambda_j}\right)|\bra{\psi_k}\ket{\phi_j}|^2\ . 
\end{equation}

\begin{example} For $f(x)=-\log x$, the quasi-relative entropy becomes the Umegaki relative entropy
$$S_{-\log}(\rho\|\sigma)=S(\rho\|\sigma)=\Tr (\rho\log\rho-\rho\log\sigma)\ . $$
\end{example}

\begin{example}
For $p\in(-1,2)$ and $p\neq 0,1$ let us take the function 
$$f_p(x):=\frac{1}{p(1-p)}(1-x^p)\ ,$$
which is {operator} convex. The quasi-relative entropy for this function is calculated to be
$$S_{f_p}(\rho|| \sigma)=\frac{1}{p(1-p)}\left(1-\Tr(\sigma^{p}\rho^{1-p})\right)\ .$$
\end{example}

\begin{example}
For $p\in(-1,1)$ take $q=1-p\in(0,2)$, the function
$$f_q(x)=\frac{1}{1-q}(1-x^{1-q}) $$
is operator convex. The quasi-relative entropy for this function is known as Tsallis $q$-entropy
$$S_q(\rho\|\sigma)= \frac{1}{1-q}\left(1-\Tr(\rho^{q}\sigma^{1-q})\right)\ .$$
\end{example}

\section{Quasi-relative entropy of entanglement}

Quasi-relative entropy of entanglement is defined as follows
$$E_f(\rho)=\min_{\sigma\in\cS}S_f(\rho\|\sigma)\ , $$
where $\cS$ is the set of separable states.

Note that the quasi-relative entropy satisfies all conditions (D1)-(D3). Condition (D1) was explicitly proved in \cite{V19-2}, or could easily be derived from Pinsker inequality \cite{HM16}. Condition (D2) is clear from the explicit form of the quasi-relative entropy (\ref{eq:formula}).

Throughout the paper consider a bipartite system $AB$ with finite-dimensional Hilbert spaces $\cH_A$ and $\cH_B$. Let $\{\ket{j}_A\}_j$ and $\{\ket{j}_B\}_j$ denote some bases for each Hilbert space (we will drop the subscript $A$ and $B$ since these are the only two systems we consider). For a probability distribution $\{p_j\}_{j}$ (i.e. $0\leq p_j\leq 1$ and $\sum_j p_j=1$) and a basis $\{\ket{j}\}_j$ of each finite-dimensional Hilbert space, denote a separable state
\begin{equation}\label{eq:sigma-prob}
\sigma(\{p_j\})=\sum_jp_j\ket{jj}\bra{jj}_{AB}\ .
\end{equation}

\begin{theorem}\label{thm:max} Let $f$ be an operator  monotone decreasing function, such that $f(1)=0$. 
For a  maximally entangled  pure state $\ket{\Psi^+}=\sum_j\frac{1}{\sqrt{d}} \ket{jj}$, the quasi-relative entropy of entanglement is reached for a state $\sigma(\{\frac{1}{d}\})=\sum_j\frac{1}{d}\ket{jj}\bra{jj}$, and becomes
$$E_f(\ket{\Psi^+}\bra{\Psi^+})= f(1/d)\ . $$
\end{theorem}
\begin{proof}
For any state $\sigma=\sum_k\mu_k\ket{\psi_k}\bra{\psi_k}$, the quasi-relative entropy is 
$$S_f(\ket{\Psi^+}\bra{\Psi^+}\|\sigma)=\sum_kf(\mu_k)|\bra{\psi_k}\ket{\Psi^+}|^2=\bra{\Psi^+}f(\sigma)\ket{\Psi^+}\ . $$
Since $f$ is convex and monotonically decreasing, we have
\begin{align}
\min_{\sigma\in\cS}S_f(\ket{\Psi^+}\bra{\Psi^+}\|\sigma) &=\min_{\sigma\in\cS}\bra{\Psi^+}f(\sigma)\ket{\Psi^+}\\
&\geq \min_{\sigma\in\cS}f\left(\bra{\Psi^+}\sigma\ket{\Psi^+}\right)\\
&\geq f\left(\max_{\sigma\in\cS}\bra{\Psi^+}\sigma\ket{\Psi^+}  \right)=f(1/d)\ .
\end{align}
The last equality is due to \cite{4H09}, for example.
All inequalities are reached for $\sigma(\frac{1}{d})$, and $S_f(\ket{\Psi^+}\bra{\Psi^+}\|\sigma(\{\frac{1}{d}\}))=f(1/d)$.
\end{proof}

The proof of the following theorem is inspired by \cite{VP98}.

\begin{theorem}\label{Entropy} Let $f$ be an operator monotone decreasing  function, such that $f(1)=0$ and $a_f=0$ in (\ref{low}). Let $\ket{\Psi}=\sum_j \sqrt{p_j}\ket{jj}$ be a pure state. If there exists a constant $H_f>0$ such that for any $p\in[0,1]$
\begin{equation}\label{eq:C-f}
\int_0^\infty p (t+p)^{-2}d\mu_f(t)=H_f\ ,
\end{equation}
where $\mu_f(t)$ is the measure in the integral representation (\ref{low}), then the quasi-relative entropy of entanglement is {reached} at a state $\sigma(\{p_j\})$ (\ref{eq:sigma-prob}), and becomes
$$E_f(\ket{\Psi}\bra{\Psi})=\sum_j p_j f(p_j)\ . $$
\end{theorem}
\begin{proof} Let  $\rho=\ket{\Psi}\bra{\Psi}$ and  $\sigma^*:=\sigma(\{p_j\})=\sum_jp_{j}\ket{jj}\bra{jj}$.
Consider a function
$$F(x,\sigma):=S_f(\rho\|(1-x)\sigma^*+x\sigma)\ . $$
For any state $\sigma$, adopting notation $\Delta(x):=\Delta_{(1-x)\sigma^*+x\sigma,\rho}$ and $\Delta^*:=\Delta_{\sigma^*,\rho}$, we have
\begin{align}
\frac{\partial F}{\partial x}(0,\sigma)=&\lim_{x\rightarrow 0}\frac{1}{x}\Tr\{ [f(\Delta(x))-f(\Delta^*)]\rho\}\\
=&\lim_{x\rightarrow 0}\frac{1}{x}\int_0^\infty d\mu_f(t)\Tr\left[\{(t\idty+\Delta(x))^{-1}-(t\idty+\Delta^*)^{-1}\}\rho\right]  \\
=&  \lim_{x\rightarrow 0}\frac{1}{x}\int_0^\infty d\mu_f(t)\Tr\left[\{(t\idty+\Delta(x))^{-1}(\Delta^*-\Delta(x))(t\idty+\Delta^*)^{-1}\}\rho\right]  \\
=& \int_0^\infty d\mu_f(t)\Tr\left[\{(t\idty+\Delta^*)^{-1}(\sigma^*-\sigma)(t\idty+\Delta^*)^{-1}\}(\rho)\right]  \\
\end{align}
The formula $A^{-1}-B^{-1}=A^{-1}(B-A)B^{-1}$ holds for any invertible operators $A$ and $B$. 

For $\rho$ and $\sigma^*$ defined above we have
$$\Delta^*=\sum_{j}p_{j}P_{j}\ , \qquad P_{j}(X)=\bra{jj}X\ket{\Psi}\ \ket{jj}\bra{\Psi}\ .$$ 
Then
$$(t\idty+\Delta^*)^{-1}(\rho)= \sum_{j}(t+p_{j})^{-1}P_{j}(I)=\sum_{j}(t+p_{j})^{-1}\bra{jj}\ket{\Psi}\ket{jj}\bra{\Psi}\ .$$
And therefore,
$$\{(t\idty+\Delta^*)^{-1}(\sigma^*-\sigma)(t\idty+\Delta^*)^{-1}\}(\rho)=\sum_{jk} (t+p_{k})^{-1}(t+p_{j})^{-1}\bra{jj}\ket{\Psi}\ket{kk}\bra{\Psi}\bra{kk}(\sigma^*-\sigma)\ket{jj}\bra{\Psi}  \ket{\Psi}\ . $$
Denote
$$g_t(p,q):= \sqrt{pq}\, (t+p)^{-1}(t+q)^{-1} \ .$$Then taking the trace above, we obtain
\begin{align}
&\Tr\left[\{(t\idty+\Delta^*)^{-1}(\sigma^*-\sigma)(t\idty+\Delta^*)^{-1}\}(\rho)\right]\nonumber\\
&=\sum_{jk} (t+p_{k})^{-1}(t+p_{j})^{-1}\bra{jj}\ket{\Psi}\bra{\Psi}\ket{kk}\bra{kk}(\sigma^*-\sigma)\ket{jj}\\
&=\sum_{jk}g_t(p_j,p_k)\bra{kk}(\sigma^*-\sigma)\ket{jj}\\
&=\sum_{j}p_j^2 (t+p_{j})^{-2}-\sum_{jk}g_t(p_j,p_k)\bra{kk}\sigma\ket{jj}\ .
\end{align}
Therefore,
\begin{align}
\frac{\partial F}{\partial x}(0,\sigma)=& \int_0^\infty d\mu_f(t)\sum_{j}p_j^2 (t+p_{j})^{-2}-\sum_{jk}g_t(p_j,p_k)\bra{kk}\sigma\ket{jj}\\
=&\,  \sum_{j} \int_0^\infty p_j^2 (t+p_{j})^{-2}d\mu_f(t)-\int_0^\infty \Tr\left(\sigma \ \sum_{jk}g_t(p_j,p_k)\ket{jj}\bra{kk}\right) d\mu_f(t)\\
\end{align}
Define
$$G_f(p,q):=\int_0^\infty g_t(p,p)d\mu_f(t)=\int_0^\infty \sqrt{pq} (t+p)^{-1}(t+q)^{-1} d\mu_f(t) \ ,$$
and
$$ H_f(p):=G_f(p,p)=\int_0^\infty p (t+p)^{-2}d\mu_f(t)\ .$$
Since
$(t+p)(t+q)\geq (\sqrt{pq}+t)^2, $
we have
$$0\leq G_f(p,q)=\int_0^\infty \sqrt{pq} (t+p)^{-1}(t+q)^{-1} d\mu_f(t)\leq \int_0^\infty \sqrt{pq} (t+\sqrt{pq})^{-2}d\mu_f(t)=H_f(\sqrt{pq})\ .  $$
Take $\sigma=\ket{\alpha\beta}\bra{\alpha\beta}$, where $\ket{\alpha}=\sum_ja_j\ket{j}$ and $\ket{\beta}=\sum_jb_j\ket{j}$. Then
\begin{align}
\frac{\partial F}{\partial x}(0,\ket{\alpha\beta}\bra{\alpha\beta})=&\sum_j p_jH_f(p_j)-\sum_{jk}G_f(p_j, p_k)\bra{kk}\sigma\ket{jj}\\
=&\sum_j p_jH_f(p_j)-\sum_{jk}G_f(p_j, p_k)\bra{kk}\ket{ab}\bra{ab}\ket{jj}\\
=&\sum_j p_jH_f(p_j)-\sum_{jk}G_f(p_j, p_k)a_kb_k\overline{a_j}\overline{b_j}
\end{align}
We assumed that the function $f$ is such that $H_f(p)=H_f$ is independent of $p$, then
$$\frac{\partial F}{\partial x}(0,\ket{\alpha\beta}\bra{\alpha\beta})=H_f-\sum_{jk}G_f(p_j, p_k)a_kb_k\overline{a_j}\overline{b_j}\ . $$
And therefore,
$$\left| \frac{\partial F}{\partial x}(0,\ket{\alpha\beta}\bra{\alpha\beta})-H_f\right|\leq \sum_{jk}G_f(p_j, p_k) |a_j||b_j| |a_k||b_k|\leq  \sum_{jk}H_f(\sqrt{p_j p_k}) |a_j||b_j| |a_k||b_k|\leq H_f\ .$$
Since every $\sigma\in \cS$ can be written as a convex combination $\sigma=\sum_jq_j\ket{\alpha_j\beta_j}\bra{\alpha_j\beta_j}$, we obtain a non-negative partial derivate
$ \frac{\partial F}{\partial x}(0,\sigma)\geq 0.$
Moreover, quasi-relative entropy of entanglement for a pure state is calculated to be
$$E_f(\ket{\Psi}\bra{\Psi})=S_f(\ket{\Psi}\bra{\Psi}\|\sigma^*)=\Tr(f(\Delta^*)\ket{\Psi}\bra{\Psi})=\sum_jp_jf(p_j)\ . $$
\end{proof}

\begin{theorem}\label{thm:qubit}
Let $\ket{\Psi}=\sqrt{p}\ket{00}+\sqrt{1-p}\ket{11}$. Let $f$ be an operator monotone decreasing function $f$, such that $f(1)=0$ and $a_f=0$ in (\ref{low}). Let $f$ be such that for any $p\in[0,1]$ the function
$$H_f(p)=\int_0^\infty p (t+p)^{-2}d\mu_f(t)\ , $$
is either monotonically increasing or decreasing in $p$ (here $\mu_f(t)$ is the measure in the integral representation (\ref{low})). Then, the closest separable state to $\ket{\Psi}$, when measured by the quasi-relative entropy $S_f$, is  in a form $\sigma(\{q, 1-q\})$ (\ref{eq:sigma-prob}) for some $0<q< 1$. In this case, the quasi-relative entropy becomes, 
$$E_f(\ket{\Psi}\bra{\Psi})=\sum_jp_jf(q_j)\ . $$
\end{theorem}
\begin{proof} Let us denote $\sigma^*=\sigma(\{q, 1-q\})=q\ket{00}\bra{00}+(1-q)\ket{11}\bra{11}$.
From a  proof of Theorem \ref{Entropy}, we have 
$$\Delta^*=\sum_{j}q_{j}P_{j}\ , \qquad P_{j}(X)=\bra{jj}X\ket{\Psi}\ \ket{jj}\bra{\Psi}\ .$$ 
Then denoting $\rho=\ket{\Psi}\bra{\Psi}$, 
$$(t\idty+\Delta^*)^{-1}(\rho)= \sum_{j}(t+q_{j})^{-1}P_{j}(I)=\sum_{j}(t+q_{j})^{-1}\bra{jj}\ket{\Psi}\ket{jj}\bra{\Psi}\ .$$
And therefore, for any state $\sigma$, we have
$$\{(t\idty+\Delta^*)^{-1}(\sigma^*-\sigma)(t\idty+\Delta^*)^{-1}\}(\rho)=\sum_{jk} (t+q_{k})^{-1}(t+q_{j})^{-1}\bra{jj}\ket{\Psi}\ket{kk}\bra{\Psi}\bra{kk}(\sigma^*-\sigma)\ket{jj}\bra{\Psi}  \ket{\Psi}\ . $$
And taking the trace, we obtain
\begin{align}
&\Tr\left[\{(t\idty+\Delta^*)^{-1}(\sigma^*-\sigma)(t\idty+\Delta^*)^{-1}\}(\rho)\right]\nonumber\\
&=\sum_{jk} (t+q_{k})^{-1}(t+q_{j})^{-1}\bra{jj}\ket{\Psi}\bra{\Psi}\ket{kk}\bra{kk}(\sigma^*-\sigma)\ket{jj}\\
&=\sum_{jk}g_t(p_j,p_k)\bra{kk}(\sigma^*-\sigma)\ket{jj}\\
&=\sum_{j}p_jq_j (t+q_{j})^{-2}-\sum_{jk}g_t(p_j,p_k)\bra{kk}\sigma\ket{jj}\ .
\end{align}
Therefore,
\begin{align}
\frac{\partial F}{\partial x}(0,\sigma)=& \int_0^\infty d\mu_f(t)\sum_{j}p_jq_j (t+q_{j})^{-2}-\sum_{jk}g_t(p_j,p_k)\bra{kk}\sigma\ket{jj}\\
=&\,  \sum_{j} \int_0^\infty p_jq_j (t+q_{j})^{-2}d\mu_f(t)-\int_0^\infty \Tr\left(\sigma \ \sum_{jk}\sqrt{\frac{p_jp_k}{q_jq_k}}g_t(p_j,p_k)\ket{jj}\bra{kk}\right) d\mu_f(t)\\
=& \sum_j p_jH_f(q_j)- \sum_{jk}\sqrt{\frac{p_jp_k}{q_jq_k}}G_f(q_j, q_k)\bra{kk}\sigma\ket{jj}\ .
\end{align}
Therefore, for $\sigma=\ket{\alpha\beta}\bra{\alpha\beta}$ with $\ket{\alpha}=\sum_ja_j\ket{j}$ and $\ket{\beta}=\sum_jb_j\ket{j}$, we obtain
$$\frac{\partial F}{\partial x}(0,\ket{\alpha\beta}\bra{\alpha\beta})=\sum_j p_jH_f(q_j)-\sum_{jk}\sqrt{\frac{p_jp_k}{q_jq_k}}G_f(q_j, q_k)a_kb_k\overline{a_j}\overline{b_j}\ . $$
For $\ket{\Psi}=\sqrt{p}\ket{00}+\sqrt{1-p}\ket{11}$, take $\sigma=\ket{00}\bra{00}$,
\begin{equation}\label{eq:partial-0}
\frac{\partial g}{\partial x}(0,\ket{00}\bra{00})=(1-p)H_f(1-q)-p\frac{1-q}{q}H_f(q) \ .
\end{equation}
And for $\sigma=\ket{11}\bra{11}$,
\begin{equation}\label{eq:partial-1}
\frac{\partial F}{\partial x}(0,\ket{11}\bra{11})=pH_f(q)-(1-p)\frac{q}{1-q}H_f(1-q)=-\frac{q}{1-q}\frac{\partial F}{\partial x}(0,\ket{00}\bra{00}) \ .
\end{equation}
Since we need both (\ref{eq:partial-0}) and (\ref{eq:partial-1}) to be non-negative, choose $q$ such that $\frac{\partial F}{\partial x}(0,\ket{00}\bra{00})=0$. In other words, choose $q$ such that
 $$\frac{1-p}{1-q}H_f(1-q)-\frac{p}{q}H_f(q)=0\ .$$
 This expression can be written as
 \begin{equation}\label{eq:p-q}
 p=\frac{qH_f(1-q)}{qH_f(1-q)+(1-q)H_f(q)}=\frac{\int_0^\infty (t+1-q)^{-2}d\mu_f(t)}{\int_0^\infty \left[(t+1-q)^{-2}+(t+q)^{-2}\right]d\mu_f(t)}\ .
 \end{equation}
Let us note a few cases here:
\begin{itemize}
\item If $H_f(p)$ is a constant, then $p=q$ for all $p\in[0,1]$, which is in line with Theorem \ref{Entropy}.
\item If $p=1/2$, then since the function $(t+x)^{-2}$ is monotone decreasing in $x$, and $\mu_f$ is a positive measure, we have $q=1/2$. And vice versa. This is in line with Theorem \ref{thm:max}.
\item If $p>1/2$, then it implies that 
$$\int_0^\infty (t+q)^{-2}d\mu_f(t)<\int_0^\infty (t+1-q)^{-2}d\mu_f(t)\ .$$ 
Since the function $(t+x)^{-2}$ is monotone decreasing in $x$, and $\mu_f$ is a positive measure, we have that $q>1-q$, or $q>1/2$. And vice versa, i.e. if $q>1/2$, then $p>1/2$.
\item If  $H_f(q)$ is an increasing function and  $p>1/2$, then we also have that  $p<q$. So, 
$$\frac{1}{2}<p<q\leq1\ . $$
And similarly, 
$$0<q<p<\frac{1}{2}\ . $$
\item Similarly, if $H_f(q)$ is a decreasing function, we have 
$$0<p<q<\frac{1}{2},\qquad  \frac{1}{2}<q<p\leq1\ .$$
\end{itemize}
Consider the difference
\begin{align}
\left|\frac{\partial F}{\partial x}(0,\ket{\alpha\beta}\bra{\alpha\beta})-\sum_j p_jH_f(q_j)\right| &=\left|\sum_{jk}\sqrt{\frac{p_jp_k}{q_jq_k}}G_f(q_j, q_k)a_kb_k\overline{a_j}\overline{b_j}\right|\\
&\leq \sum_{jk}\sqrt{\frac{p_jp_k}{q_jq_k}}G_f(q_j, q_k)|a_kb_k{a_j}{b_j}| \\
&=2\sum_{k<j}\sqrt{\frac{p_jp_k}{q_jq_k}}G_f(q_j, q_k)|a_kb_k{a_j}{b_j}|+\sum_j\frac{p_j}{q_j}H_f(q_j)|a_j|^2|b_j|^2\\
\end{align}
The last equality is due to the fact that $G_f$ is symmetric in its arguments.
Therefore, using (\ref{eq:p-q}), we obtain
\begin{align}
&\left|\frac{\partial F}{\partial x}(0,\ket{\alpha\beta}\bra{\alpha\beta})-pH_f(q)-(1-p)H_f(1-q)\right|\\
&=\left|\frac{\partial g}{\partial x}(0,\ket{\alpha\beta}\bra{\alpha\beta})-\frac{p}{q}H_f(q)\right|\\
&= 2\sqrt{\frac{p(1-p)}{q(1-q)}}G_f(q, 1-q)|a_0a_1b_0b_1|+\frac{p}{q}H_f(q)|a_0|^2|b_0|^2+\frac{1-p}{1-q}H_f(1-q)|a_1|^2|b_1|^2\\
&= 2\sqrt{\frac{p(1-p)}{q(1-q)}}G_f(q, 1-q)|a_0a_1b_0b_1|+\frac{p}{q}H_f(q)\left(|a_0|^2|b_0|^2+|a_1|^2|b_1|^2\right)\\
\end{align}
Let us denote $A:=|a_0||b_0|$ and $B:=|a_1||b_1|$. Then $(A+B)^2\leq\sum_j|a_j|^2\sum_j|b_j|^2\leq 1.$ Therefore, $B\leq 1-A$. And thus
\begin{align}
&\left|\frac{\partial F}{\partial x}(0,\ket{\alpha\beta}\bra{\alpha\beta})-\frac{p}{q}H_f(q)\right|\\
&= 2\sqrt{\frac{p(1-p)}{q(1-q)}}G_f(q, 1-q)|a_0a_1b_0b_1|+\frac{p}{q}H_f(q)\left(|a_0|^2|b_0|^2+|a_1|^2|b_1|^2\right)\\
&\leq \frac{p}{q}H_f(q)+2A(1-A)\left[\sqrt{\frac{p(1-p)}{q(1-q)}}G_f(q, 1-q)- \frac{p}{q}H_f(q)  \right]\ .
\end{align}
Note that from  (\ref{eq:p-q}) the following holds
\begin{align}
\sqrt{\frac{p(1-p)}{q(1-q)}}G_f(q, 1-q)- \frac{p}{q}H_f(q)&=\sqrt{\frac{p(1-p)}{q(1-q)}}\left[ G_f(q, 1-q)-\sqrt{\frac{p(1-q)}{q(1-p)}}H_f(q)\right]\\
&=\sqrt{\frac{p(1-p)}{q(1-q)}}\left[  G_f(q, 1-q)-\sqrt{H_f(q)H_f(1-q)}\right]\\
&\leq 0 \ .
\end{align}
The last inequality is due to the Cauchy-Swartz inequality:
\begin{align}
G_f(p,q)^2&=pq\left(\int_0^\infty (t+p)^{-1}(t+q)^{-1}d\mu_f(t)\right)^2\\
&\leq pq\int_0^\infty (t+p)^{-2}d\mu_f(t) \int_0^\infty (t+q)^{-2}d\mu_f(t) =H_f(p)H_f(q)\ .
\end{align}
Therefore, going back we obtain 
\begin{align}
&\left|\frac{\partial F}{\partial x}(0,\ket{\alpha\beta}\bra{\alpha\beta})-\frac{p}{q}H_f(q)\right|\\
&\leq \frac{p}{q}H_f(q)+2A(1-A)\left[\sqrt{\frac{p(1-p)}{q(1-q)}}G_f(q, 1-q)- \frac{p}{q}H_f(q)  \right]\\
&\leq \frac{p}{q}H_f(q)\ .
\end{align}
This means that $\frac{\partial F}{\partial x}(0,\ket{\alpha\beta}\bra{\alpha\beta})\geq 0$. Since every $\sigma\in \cS$ can be written as a convex combination $\sigma=\sum_jq_j\ket{\alpha_j\beta_j}\bra{\alpha_j\beta_j}$, we obtain a non-negative partial derivate
$ \frac{\partial F}{\partial x}(0,\sigma)\geq 0.$ Therefore, the chosen $\sigma^*$ is the closest separable state to $\ket{\Psi}$.
\end{proof}

As a consequence we obtain partial cases of Theorems discussed above:
\begin{itemize}
\item If $H_f(p)$ is a constant, then $p=q$ for all $p\in[0,1]$. And the closest separable state to a state $\ket{\Psi}=\sqrt{p}\ket{00}+\sqrt{1-p}\ket{11}$ is the state $\sigma^*=p\ket{00}\bra{00}+(1-p)\ket{11}\bra{11}$, which is a case of a two-qubit system in Theorem \ref{Entropy}.
\item If $p=1/2$, then  $q=1/2$. Therefore, the closest separable state to a maximally entangled state $\ket{\Psi}=\frac{1}{\sqrt{2}}(\ket{00}+\ket{11})$ is the state $\sigma^*=\frac{1}{2}(\ket{00}\bra{00}+\ket{11}\bra{11})$. This is a  case of a two-qubit system in Theorem \ref{thm:max}. 
\end{itemize}

\section{Examples}

\begin{example}
For $f(x)=-\log(x)$, the quasi-relative entropy becomes regular Umegaki relative entropy
$$S_{-log}(\rho\|\sigma)=S(\rho\|\sigma)=\Tr(\rho\log\rho-\rho\log\sigma)\ . $$
For this function, $H_f(p)=1$ for all $p\in[0,1]$.
As it was shown in \cite{VP98} the relative entropy of entanglement reduces to the marginal entropy on pure states, coinciding with Theorem \ref{Entropy},
$$E(\ket{\Psi}\bra{\Psi})=S(\rho_A)=-\sum_j p_j \log p_j\ , $$
where $\rho_A=\Tr_B \ket{\Psi}\bra{\Psi}=\sum_j p_j\ket{j}\bra{j}.$ This holds for any pure state $\rho=\ket{\Psi}\bra{\Psi}$.
\end{example}

\begin{example}
The power function $f(x)=1-x^{1-\alpha}$ for $\alpha\in(0,1)$ defines the quasi-relative entropy as
$$S_\alpha(\rho\|\sigma):=1-\Tr(\rho^\alpha\sigma^{1-\alpha})\ . $$
By Theorem \ref{thm:max}, for a maximally entangled state $\ket{\Psi^+}=\frac{1}{\sqrt{d}}\sum_j\ket{jj}$, the minimum of the $\alpha$-relative entropy over separable states is
$$ \min_{\sigma\in\cS}S_\alpha \left(\ket{\Psi^+}\bra{\Psi^+}\|\sigma\right)=1-d^{\alpha-1}\ .$$

\end{example}

\begin{example}
Renyi entropy of entanglement is defined as
$$E_\alpha^R(\rho)=\min_{\sigma\in \cS} S^R_\alpha(\rho\|\sigma)\ , $$
for $\alpha\geq 0$, $\alpha\neq 1$, where the Renyi relative entropy is defined as
$$S^R_\alpha(\rho\|\sigma)= \frac{1}{\alpha-1}\log\Tr(\rho^\alpha\sigma^{1-\alpha})=\frac{1}{\alpha-1}\log\left(1-S_\alpha(\rho\|\sigma) \right)\ .$$
Consider the case when  $\alpha<1$. Then
$$E_\alpha^R(\rho)= \frac{1}{\alpha-1}\log[1-\min_{\sigma\in \cS} S_\alpha(\rho\|\sigma)]\ . $$
And therefore, for a maximally entangled  state, we have
$$E_\alpha^R(\ket{\Psi^+}\bra{\Psi^+})= \frac{1}{\alpha-1}\log d^{\alpha-1}=\log d . $$
This means that for a maximally entangled state, the Renyi entropy of entanglement reduces to the marginal Renyi entropy, since
$$S_\alpha^R(\rho_A)=\frac{1}{1-\alpha}\log\Tr\rho_A^\alpha=\log d=E_\alpha^R(\ket{\Psi^+}\bra{\Psi^+})\ , $$
where $\rho_A=\Tr_B\ket{\Psi^+}\bra{\Psi^+}=\frac{1}{d}\sum_j\ket{jj}\bra{jj}.$ 
This is in line with the explicit expression found for the  Renyi entropy of entanglement in \cite{CVGG20, ZHC17}. Note that in \cite{CVGG20} it was shown that for a pure state $\ket{\Psi}=\sum_j\sqrt{p_j}\ket{jj}$ and $\alpha<2$, the closest separable state is 
$$\sigma=\frac{1}{P_\alpha}\sum_j p_j^{1/\alpha} \ket{jj}\bra{jj}\ ,$$
where $P_\alpha=\sum_j p_j^{1/\alpha}.$ The Renyi entropy of entanglement is then given by
$$E_\alpha^R(\Psi)=\frac{\alpha}{\alpha-1}\log\sum_j p_j^{1/\alpha}\ . $$ 
This would be a generalization of Theorem \ref{thm:qubit} in Renyi relative entropy case, although more work is needed to explicitly show that $H_f(p)$ is monotonic in this case, and to explicitly find coefficients in the separable state to apply the Theorem for a general pure state.

\end{example}

\begin{example}
Tsallis entropy of entanglement is defined as
$$E_\alpha^T(\rho)=\min_{\sigma\in \cS} S^T_\alpha(\rho\|\sigma)\ , $$
for  $\alpha\geq 0$, $\alpha\neq 1$, where Tsallis relative entropy is defined as
$$S_\alpha^T(\rho\|\sigma)=\frac{1}{1-\alpha}\left(1-\Tr(\rho^\alpha\sigma^{1-\alpha}) \right)=\frac{1}{1-\alpha} S_\alpha(\rho\|\sigma) \ .$$
By Theorem \ref{thm:max}, for $\alpha<1$, and a maximally entangled  state, we have
$$E_\alpha^T(\ket{\Psi^+}\bra{\Psi^+})=\frac{1}{1-\alpha}\min_{\sigma\in\cS}S_\alpha(\ket{\Psi^+}\bra{\Psi^+}\|\sigma)=\frac{1}{1-\alpha}\left[ 1-d^{\alpha-1}\right]\ .$$
Interestingly enough, Tsallis entropy of entanglement does not reduce to the marginal Tsallis entropy on a maximally entangled state, since
$$S_\alpha^T(\rho_A)=\frac{1}{\alpha-1}(1-\Tr\rho_A^\alpha)=\frac{1}{\alpha-1}\left[ 1-d^{1-\alpha}\right]\ , $$
where $\rho_A=\Tr_B\ket{\Psi^+}\bra{\Psi^+}=\frac{1}{d}\sum_j\ket{jj}\bra{jj}.$

\end{example}

\section{Conclusion}
We asked a question: what is the closest separable state to a pure state when the distance between two states is measured by the quasi-relative entropy? We answered this question in three cases: when a state is maximally entangled; when a state is on a bipartite qubit systems with a large class of functions defining quasi-relative entropy; and for any pure state and a certain class of functions. It would be interesting to generalize a qubit case (Theorem \ref{thm:qubit}) to any finite dimension and/or remove the assumptions on the functions $f$, which would result in a complete answer of the initial question.

There are multiple directions one could take to further develop this problem. Let us a list just a few possibilities, besides the direct improvement of the presented results.

\begin{itemize}
\item It is important to answer the same question for a variety of other entanglement measures, not obtained from the quasi-relative entropy of entanglement. These measures can be induced, for example, by sandwiched relative entropy \cite{MLDSFT13, WWY14}, $\alpha-z$-Renyi relative entropy \cite{AD15, CFL18, Z20}, geometric Renyi divergence \cite{FF19, M15}, to name a few.

\item A more general question of explicitly finding a closest separable state for a mixed bipartite entangled stats is still open. The way to find an answer for two-qubit systems has been discussed in \cite{MI08} for some restricted cases. For a review on the problem, see \cite{KW05}. In \cite{BV04, DPS02, EHGC04, GZFG15, ZFG10} authors provide an algorithm for the calculation of the relative entropy of entanglement and for the separability problem. The existence of a similar algorithm could be explored for a quasi-relative entropy.

\item An inverse problem has be considered for a while: given a separable state, which entangled states are the closest to it? For two qubit states a closed formula for entangled states has been found in \cite{MI08}, which was generalized to any dimensions and any number of parties in \cite{FG11}. A natural question arises to generalize this result for quasi-relative entropy.

\item As mentioned in Introduction, there are a number of entanglement measure for a mixed state. For a review see, for example, \cite{3H00}. It is interesting to note that if an entanglement of one state is less than an entanglement of another state when measured by, say, relative entropy, it is not necessarily the case that this order will be preserved using other measures. In fact it was demonstrated in \cite{VP00} that all  asymptotic entanglement measures satisfying (E1)-(E3) are either identical or have a different ordering on the set of all quantum states. So the following problem arises: for a fixed entanglement value of one kind of entanglement measure, find a state with a maximum or minimum value of entanglement measured by another. Some numerical results have been done, for example, in \cite{BHLM13, MG04, MIHG08}. Adding quasi-relative entropy into the mix, poses new frontiers in this problem.

\end{itemize}

\vspace{0.3in}
\textbf{Acknowledgments.} The author would like to thank anonymous referees who worked diligently to make this manuscript better on multiple levels.  A. V. is supported by NSF grant DMS-1812734.


\begin{thebibliography}{10}

\bibitem{A76} Araki, H. (1976). Relative entropy of states of von Neumann algebras. Publications of the Research Institute for Mathematical Sciences, 11(3), 809-833.

\bibitem{AD15} Audenaert, K. M.,  Datta, N. (2015). $\alpha-z$-Renyi relative entropies. Journal of Mathematical Physics, 56(2), 022202.

\bibitem{BHLM13} Bartkiewicz, K., Horst, B., Lemr, K., Miranowicz, A. (2013). Entanglement estimation from Bell inequality violation. Physical Review A, 88(5), 052105.

\bibitem{BB84} Bennett, C. H.,  Brassard, G. (2020). Quantum cryptography: Public key distribution and coin tossing. arXiv preprint arXiv:2003.06557.

\bibitem{BBCJPW93} Bennett, C. H., Brassard, G., Crepeau, C., Jozsa, R., Peres, A., Wootters, W. K. (1993). Teleporting an unknown quantum state via dual classical and Einstein-Podolsky-Rosen channels. Physical review letters, 70(13), 1895.


\bibitem{B97} Bhatia, {\em Matrix analysis,} Springer-Verlag, New York, 1997

\bibitem{BNO02} Biham, O., Nielsen, M. A.,  Osborne, T. J. (2002). Entanglement monotone derived from Grover?s algorithm. Physical Review A, 65(6), 062312.

\bibitem{BKAB00}  Boto, A. N., Kok, P., Abrams, D. S., Braunstein, S. L., Williams, C. P.,  Dowling, J. P. (2000). Quantum interferometric optical lithography: exploiting entanglement to beat the diffraction limit. Physical Review Letters, 85(13), 2733.

\bibitem{BV04} Brandao, F. G.,  Vianna, R. O. (2004). Separable multipartite mixed states: operational asymptotically necessary and sufficient conditions. Physical review letters, 93(22), 220503.

\bibitem{CFL18} Carlen, E. A., Frank, R. L.,  Lieb, E. H. (2018). Inequalities for quantum divergences and the Audenaert-Datta conjecture. Journal of Physics A: Mathematical and Theoretical, 51(48), 483001.

\bibitem{CV-18} Carlen, E. A.,  Vershynina, A. (2018). Recovery and the Data Processing Inequality for quasi-entropies. IEEE Transactions on Information Theory, 64(10), 6929-6938.

\bibitem{CVGG20} Chitambar, E., de Vicente, J. I., Girard, M. W.,  Gour, G. (2020). Entanglement manipulation beyond local operations and classical communication. Journal of Mathematical Physics, 61(4), 042201.

\bibitem{C67-2}  Csiszar, I. (1967). Information-type measures of difference of probability distributions and indirect observation. studia scientiarum Mathematicarum Hungarica, 2, 229-318.

\bibitem{DPS02} Doherty, A. C., Parrilo, P. A., Spedalieri, F. M. (2002). Distinguishing separable and entangled states. Physical Review Letters, 88(18), 187904.

\bibitem{Dono} Donoghue Jr, W. F. (1974). Monotone Matrix Functions and Analytic Continuation, vol. 207 of. Die Grundlehren der mathematischen Wissenschaften.


\bibitem{DCGZ03}  Durt, T., Kaszlikowski, D., Chen, J. L.,  Kwek, L. C. (2004). Security of quantum key distributions with entangled qudits. Physical Review A, 69(3), 032313.

\bibitem{EHGC04} Eisert, J., Hyllus, P., Guhne, O.,  Curty, M. (2004). Complete hierarchies of efficient approximations to problems in entanglement theory. Physical Review A, 70(6), 062317.

\bibitem{FF19} Fang, K.,  Fawzi, H. (2019). Geometric R\'{e} nyi Divergence and its Applications in Quantum Channel Capacities. arXiv preprint arXiv:1909.05758.

\bibitem{FKR04} Fan, H., Korepin, V.,  Roychowdhury, V. (2004). Entanglement in a valence-bond solid state. Physical review letters, 93(22), 227203.

\bibitem{FG11} Friedland, S., Gour, G. (2011). An explicit expression for the relative entropy of entanglement in all dimensions. Journal of mathematical physics, 52(5), 052201.

\bibitem{GGF14} Girard, M. W., Gour, G., Friedland, S. (2014). On convex optimization problems in quantum information theory. Journal of Physics A: Mathematical and Theoretical, 47(50), 505302.

\bibitem{GZFG15} Girard, M. W., Zinchenko, Y., Friedland, S.,  Gour, G. (2015). Erratum: Numerical estimation of the relative entropy of entanglement [Phys. Rev. A 82, 052336 (2010)]. Physical Review A, 91(2), 029901.


\bibitem{HM16} Hiai, F.,  Mosonyi, M. (2017). Different quantum $f$-divergences and the reversibility of quantum operations. Reviews in Mathematical Physics, 29(07), 1750023.

\bibitem{3H98} Horodecki, M., Horodecki, P.,  Horodecki, R. (1998). Mixed-state entanglement and distillation: Is there a ``bound" entanglement in nature?. Physical Review Letters, 80(24), 5239.

\bibitem{3H00} Horodecki, M., Horodecki, P., Horodecki, R. (2000). Limits for entanglement measures. Physical Review Letters, 84(9), 2014.

\bibitem{4H09} Horodecki, R., Horodecki, P., Horodecki, M.,  Horodecki, K. (2009). Quantum entanglement. Reviews of modern physics, 81(2), 865.





\bibitem{LZYZ03}  Jing, J., Zhang, J., Yan, Y., Zhao, F., Xie, C., Peng, K. (2003). Experimental demonstration of tripartite entanglement and controlled dense coding for continuous variables. Physical review letters, 90(16), 167903.

\bibitem{KW05} Krueger, O.,  Werner, R. F. (2005). Some open problems in quantum information theory. arXiv preprint quant-ph/0504166.

\bibitem{LSADSS07} Lewenstein, M., Sanpera, A., Ahufinger, V., Damski, B., Sen, A.,  Sen, U. (2007). Ultracold atomic gases in optical lattices: mimicking condensed matter physics and beyond. Advances in Physics, 56(2), 243-379.

\bibitem{LWGZ08}  Lin, S., Wen, Q. Y., Gao, F.,  Zhu, F. C. (2008). Quantum secure direct communication with $\xi$-type entangled states. Physical Review A, 78(6), 064304.

\bibitem{M15} Matsumoto, K. (2015). A new quantum version of f-divergence. In Nagoya Winter Workshop: Reality and Measurement in Algebraic Quantum Theory (pp. 229-273). Springer, Singapore.

\bibitem{MG04} Miranowicz, A., Grudka, A. (2004). A comparative study of relative entropy of entanglement, concurrence and negativity. Journal of Optics B: Quantum and Semiclassical Optics, 6(12), 542.

\bibitem{MI08} Miranowicz, A., Ishizaka, S. (2008). Closed formula for the relative entropy of entanglement. Physical Review A, 78(3), 032310.

\bibitem{MIHG08} Miranowicz, A., Ishizaka, S., Horst, B., Grudka, A. (2008). Comparison of the relative entropy of entanglement and negativity. Physical Review A, 78(5), 052308.

\bibitem{MLDSFT13} Muller-Lennert, M., Dupuis, F., Szehr, O., Fehr, S.,  Tomamichel, M. (2013). On quantum R\'enyi entropies: A new generalization and some properties. Journal of Mathematical Physics, 54(12), 122203.


\bibitem{NKL98}  Nielsen, M. A., Knill, E.,  Laflamme, R. (1998). Complete quantum teleportation using nuclear magnetic resonance. Nature, 396(6706), 52-55.

\bibitem{PZFFLG03}  Peng, X., Zhu, X., Fang, X., Feng, M., Liu, M.,  Gao, K. (2003). Experimental implementation of remote state preparation by nuclear magnetic resonance. Physics Letters A, 306(5-6), 271-276.


\bibitem{P85} Petz, D. (1985). Quasi-entropies for states of a von Neumann algebra. Publications of the Research Institute for Mathematical Sciences, 21(4), 787-800.



\bibitem{P86} Petz, D. (1986). Quasi-entropies for finite quantum systems. Reports on mathematical physics, 23(1), 57-65.


\bibitem{PV14} Plenio, M. B.,  Virmani, S. S. (2014). An introduction to entanglement theory. In Quantum Information and Coherence (pp. 173-209). Springer, Cham.

\bibitem{SSB06} Shapira, D., Shimoni, Y.,  Biham, O. (2006). Groverian measure of entanglement for mixed states. Physical Review A, 73(4), 044301.


\bibitem{SKB10} Streltsov, A., Kampermann, H.,  Bruss, D. (2010). Linking a distance measure of entanglement to its convex roof. New Journal of Physics, 12(12), 123004.

\bibitem{SWGZ09} Sun, Y., Wen, Q. Y., Gao, F.,  Zhu, F. C. (2009). Robust variations of the Bennett-Brassard 1984 protocol against collective noise. Physical Review A, 80(3), 032321.

\bibitem{TUU18}Takayanagi, T., Ugajin, T.,  Umemoto, K. (2018). Towards an entanglement measure for mixed states in CFTs based on relative entropy. Journal of High Energy Physics, 2018(10), 166.

\bibitem{VP98} Vedral, V.,  Plenio, M. B. (1998). Entanglement measures and purification procedures. Physical Review A, 57(3), 1619.

\bibitem{VPRK97} Vedral, V., Plenio, M. B., Rippin, M. A.,  Knight, P. L. (1997). Quantifying entanglement. Physical Review Letters, 78(12), 2275.

\bibitem{V19} Vershynina, A. (2019). Entanglement rates for Renyi, Tsallis, and other entropies. Journal of Mathematical Physics, 60(2), 022201.

\bibitem{V19-2} Vershynina, A. (2019). On quantum quasi-relative entropy. Reviews in Mathematical Physics, 31(07), 1950022.

\bibitem{V19-3} Vershynina, A. (2019). Upper continuity bound on the quantum quasi-relative entropy. Journal of Mathematical Physics, 60(10), 102201.

\bibitem{VP00} Virmani, S., Plenio, M. B. (2000). Ordering states with entanglement measures. Physics Letters A, 268(1-2), 31-34.

\bibitem{V16} Virosztek, D. (2016). Quantum entropies, relative entropies, and related preserver problems.

\bibitem{WWY14} Wilde, M. M., Winter, A.,  Yang, D. (2014). Strong converse for the classical capacity of entanglement-breaking and Hadamard channels via a sandwiched R\'enyi relative entropy. Communications in Mathematical Physics, 331(2), 593-622.

\bibitem{Z20} Zhang, H. (2020). From Wigner-Yanase-Dyson conjecture to Carlen-Frank-Lieb conjecture. Advances in Mathematics, 365, 107053.

\bibitem{ZHC17} Zhu, H., Hayashi, M.,  Chen, L. (2017). Coherence and entanglement measures based on R\'enyi relative entropies. Journal of Physics A: Mathematical and Theoretical, 50(47), 475303.

\bibitem{ZFG10} Zinchenko, Y., Friedland, S., Gour, G. (2010). Numerical estimation of the relative entropy of entanglement. Physical Review A, 82(5), 052336.


\end{thebibliography}
\end{document}